\documentclass[runningheads, envcountsame, a4paper]{llncs}

\usepackage{amssymb}
\usepackage{amsmath}
\usepackage{mathrsfs}
\usepackage{microtype}
\usepackage[hyphens]{url}
\usepackage{graphicx}

\usepackage[unicode]{hyperref}
\hypersetup{colorlinks=true, linkcolor=blue}

\usepackage{mytikz} 

\spnewtheorem*{definition*}{Definition}{\bfseries\upshape}{\itshape}

\newcommand{\N}{\mathbb{N}}
\newcommand{\gA}{\mathfrak{A}}
\newcommand{\Sym}{\mathbf{S}}
\newcommand{\Tran}{\mathbf{T}}
\newcommand{\Aut}{\mathscr{A}}
\newcommand{\cB}{\mathcal{B}}
\newcommand{\id}{e}

\newcommand{\abs}[1]{{\lvert#1\rvert}}
\newcommand{\floor}[1]{{\left\lfloor{#1}\right\rfloor}}
\newcommand{\ceil}[1]{{\left\lceil{#1}\right\rceil}}
\newcommand{\len}[1]{{\ell(#1)}}

\newcommand{\tranf}{\delta}
\newcommand{\outf}{\lambda}

\newcommand{\eword}{\varepsilon}

\newcommand{\bydef}{=}
\begin{document}

\mainmatter  % start of an individual contribution

\title{Preset Distinguishing Sequences\\ and\\ Diameter of Transformation Semigroups}

\titlerunning{Preset Distinguishing Sequences and Diameter of Transformation Semigroups}
\toctitle{Preset Distinguishing Sequences and Diameter of Transformation Semigroups}

\author{Pavel Panteleev}
\authorrunning{P.~Panteleev} 
\tocauthor{P.~Panteleev}

\institute{Faculty of Mechanics and Mathematics,\\ 
Lomonosov Moscow State University,\\
GSP-1, Leninskiye Gory, Moscow, 119991, Russian Federation
}
\maketitle

\begin{abstract}
 We investigate the length $\ell(n,k)$ of a shortest preset distinguishing sequence (PDS)  in the worst case for a~$k$-element subset of an~$n$-state Mealy automaton. It was mentioned by Sokolovskii \cite{sokolovskii:1976:en} that this problem is closely related to the problem of finding the maximal subsemigroup diameter $\ell(\mathbf{T}_n)$ for the full transformation semigroup $\mathbf{T}_n$ of an $n$-element set. We prove that $\ell(\mathbf{T}_n)=2^n\exp\{\sqrt{\frac{n}{2}\ln n}(1+ o(1))\}$ as $n\to\infty$ and, using approach of Sokolovskii, find the asymptotics of $\log_2 \ell(n,k)$ as $n,k\to\infty$ and $k/n\to a\in (0,1)$.

\end{abstract}

\section{Introduction}

Finite state machines are widely used models for systems in a variety of
areas, including sequential circuits~\cite{Kohavi:1970} and communication protocols~\cite{lee:1996}. The study of finite automata testing is motivated by applications in the verification of these systems. One of the basic tasks in the verification of finite automata is to identify the state of the automaton under investigation. Once the state is known, the behavior of the automaton becomes predictable and it is possible to force the automaton into the desirable mode of operation. Suppose we have a finite deterministic Mealy automaton $\gA$ whose transition and output functions are available and we know that its initial state $q_0$ is in some subset $S$ of its set of states $Q$. The \emph{state-identification} problem is to find an input sequence called a \emph{preset distinguishing sequence} (PDS) for $S$ in $\gA$ that produces different outputs for different states from $S$. Before we give a formal definition of a PDS and state the results of the paper we need to fix notations and recall some standard definitions from automata theory. 

A \emph{finite deterministic Mealy automaton} (an \emph{automaton} for short) is a quintuple
$\gA=(A,Q,B,\tranf,\outf)$, where: $A$, $Q$, $B$ are finite
nonempty sets called the \emph{input alphabet}, the \emph{set of states}, and the \emph{output alphabet}, respectively; $\tranf\colon Q\times A\rightarrow Q$ and $\outf\colon Q\times A\rightarrow B$
are total functions called the \emph{transition function} and the \emph{output
function}, respectively.

If we omit in the definition of automaton the output alphabet~$B$ and the output function~$\outf$ we obtain an object ${\gA=(A,Q,\tranf)}$ called \emph{finite semiautomaton}. If the output function $\tranf$ is partial then the semiautomaton is also called \emph{partial}.

Let $\Sigma$ be an arbitrary alphabet. By $\Sigma^*$ we denote the set of
all words over the alphabet $\Sigma$. Denote by $|\alpha|$ the length
of a word $\alpha\in \Sigma^*$. Denote by $\eword$ the \emph{empty} word, i.e., $\abs{\eword}=0$.

As usual, we extend functions $\tranf$ and $\outf$ to the set $Q\times A^*$ in the
following way:
$\tranf(q,\eword)\bydef q$, $\tranf(q,\alpha
a) \bydef \tranf(\tranf(q,\alpha),a)$,
$\outf(q,\eword) \bydef \eword$,
$\outf(q,\alpha a) \bydef \outf(q,\alpha)\outf(\tranf(q,\alpha),a)$,
where $q\in Q, a\in A,
\alpha\in A^{*}$.
Moreover, if $S\subseteq Q$ is a subset of states, then we let $\tranf(S, \alpha) \bydef \{\delta(q, \alpha) \mid q\in S\}$.

We say that two states $q_1,q_2\in Q$ of an automaton
$\gA$ are \emph{distinguishable} by an input word $\alpha\in
A^{*}$ if $\outf(q_1,\alpha)\neq \outf(q_2,\alpha)$. If there are no such words we say that the states $q_1,q_2$ are \emph{indistinguishable} or \emph{equivalent}. An automaton is called \emph{reduced} or \emph{minimal} if it does not have equivalent states.
\begin{definition*}
Let $S$ be a subset of states of an automaton $\gA$. We say that an input word $\alpha$ is a preset distinguishing sequence (PDS) for  $S$ in $\gA$ if $\alpha$ pairwise distinguishes the states in the set $S$, i.e., $\outf(q_1,\alpha)\ne\outf(q_2,\alpha)$ for all $q_1,q_2\in S$, $q_1\ne q_2$.
\end{definition*}

Denote by $\ell(\gA,S)$ the length of a shortest PDS for $S$ in $\gA$, or $0$ if such a PDS does not exist. It is a well known fact~\cite{moore:1956:en} that there are reduced automata that do not have a~PDS for some $k$-element subsets of states when $k \ge 3$. Moreover it can be easily verified that the reduced automaton on Fig. \ref{fg:pds_example} does not have a~PDS for \emph{any} $3$ element subset of states. 

\begin{figure}\center
\begin{tikzpicture}[std-opt,auto=right,on grid,bend angle=30,every node/.style={font=\footnotesize}]
    \tikzstyle{my-state}=[state,minimum size=6mm,inner sep=1pt]
    \radius=1.7cm
    %\draw[help lines] circle(\radius);
    % States
    \foreach \angle/\state in {90/q_1,135/q_2, 270/q_i, 45/q_n}
        \node[my-state] (\state) at (\angle:\radius) {$\state$};
    \node[circle,rotate=90+202.5] (dots1) at (202.5:\radius) {$\cdots$};
    \node[circle,rotate=90+337.5] (dots2) at (337.5:\radius) {$\cdots$};
    % Transitions
    \foreach \from/\to in {q_1/q_2, q_2/dots1, dots1/q_i, q_i/dots2, dots2/q_n, q_n/q_1}
        \path (\from) edge [->,bend right=22,looseness=0.8] node {$0/0$} (\to);
    \path (q_2) edge[->,bend right=22,looseness=0.8] node[pos=0.2]       {$1/0$} (q_1);
    \path (q_n) edge[->,bend left=22,looseness=0.8]  node[pos=0.2,swap]  {$1/0$} (q_1);
    \path (q_i) edge[->]                             node                {$1/0$} (q_1);
    \path (q_1) edge[->,wide loop above]             node                {$1/1$} (q_1);
\end{tikzpicture}
\caption{The reduced automaton $\gA$ that does not have a PDS for any $3$-element state subset }\label{fg:pds_example}
\end{figure}
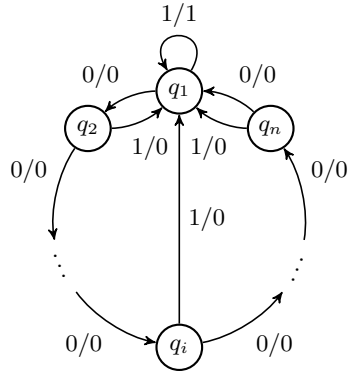 

 Consider the function
\[\ell(n,k) \bydef \max_{\gA\in\Aut_n, \abs{S}=k} \ell(\gA,S),\]
where $\Aut_n$ is the class of all $n$-state automata. This function can be interpreted as the length of a shortest~PDS in the worst case for a $k$-element subset of states in an $n$-state automaton.

The function $\ell(n,k)$ was studied by many authors. In his seminal paper \cite{moore:1956:en} Moore proves that $\ell(n,2)=n-1$. Gill~\cite{Gill:1961} gives the upper bound $\ell(n,k) \le (k-1)n^k$. Sokolovskii finds the lower bounds in \cite{sokolovskii:1971:en}:
\begin{align}
\ell(n,k) & \ge  \binom{n-1}{k-1}\text{ if }1\le k\le n/2,\label{th:sok:low1}\\
\ell(n,k) & \ge  \binom{n-2}{\floor{(n-2)/2}}\text{ if }n/2 < k < n\label{th:sok:low2}.
\end{align}

In~\cite{rystsov:1980:en} Rystsov shows that $\log_3 \ell(n,n)\sim n/6$ as $n\to\infty$. The result is proved reducing the problem of estimating $\ell(n,n)$ to the problem of estimating the function $T(n)$ that is equal to the length of a shortest irreducible word in the worst case for a partial $n$-state semiautomaton. An \emph{irreducible word} for a partial semiautomaton $\gA=(A,Q,\tranf)$ is a word $\alpha\in A^*$ such that its action is defined on all states and for any word $\beta\in A^*$ such that its action is defined on the set $\tranf(Q,\alpha)$ we have $\abs{\tranf(Q,\alpha)} = \abs{\tranf(Q,\alpha\beta)}$. In~\cite{rystsov:1980:en} it is proved that $\log_3 T(n)\sim n/3$ as $n\to\infty$. It is interesting to note that $T(n)$ coincides with the function $d_3(n)$  studied by several authors~\cite{Martyugin:2010,Gazdag:2009} which is equal to the length of a shortest carefully synchronizing word\footnote{A word $\alpha$ is a \emph{carefully synchronizing} for a partial semiautomaton $\gA=(A,Q,\tranf)$ if the action of $\alpha$ is defined on all states and $\abs{\tranf(Q,\alpha)} = 1$.} in the worst case for a partial $n$-state semiautomaton. This is due to the fact that every carefully synchronizing word is also irreducible and the worst case irreducible word is always carefully synchronizing\footnote{If $\alpha$ is a shortest irreducible word for a partial semiautomaton $\gA=(A,Q,\tranf)$ and $\abs{\tranf(Q,\alpha)}>1$ then we can always add a new input symbol $a$ to $\gA$ and obtain $\gA'=(A\cup\{a\},Q,\tranf')$ such that $\alpha a$ is a shortest irreducible word for $\gA'$ and $\abs{\tranf'(Q,\alpha a)}=1$.}. Thus we have $\log_3 d_3(n)\sim n/3$ which was conjectured in \cite{Gazdag:2009}.  

In the paper \cite{sokolovskii:1976:en} Sokolovskii investigate the relationship between the function $\ell(n,k)$ and the maximum of a subsemigroup diameter in the full transformation semigroup of an $n$-element set.

\begin{definition*} Let $\Omega_n$ be an $n$-element set. The full transformation semigroup of $\Omega_n$ (also called the symmetric semigroup of $\Omega_n$) is the set $\Tran_n$ of all transformations of $\Omega_n$.
\end{definition*}

The set $\Tran_n$ contains the proper subset $\Sym_n$ of all bijections on the set $\Omega_n$ called the \emph{symmetric group} on $\Omega_n$. We see that $\Tran_n$ is a monoid and $\Sym_n$ is a group with function composition as the multiplication operation. In this paper, by the \emph{composition} $fg$ of transformations ${f,g\in\Tran_n}$ we mean the \emph{left} composition ${x\mapsto g(f(x))}$.

Consider $\cB\subseteq \Tran_n$. By $\langle \cB\rangle$ denote the \emph{closure} of the set $\cB$, i.e., the set $\{f_1 \ldots f_\ell \mid f_1,\ldots,f_\ell\in \cB\}$. Let $f\in \langle \cB \rangle$ and $\ell$ be the minimum natural number such that $f=f_1 \ldots f_\ell$ for some $f_1,\ldots,f_\ell\in\cB$. Then $\ell$ is called the \emph{complexity} of the function $f$ over the \emph{basis} $\cB$ and is denoted by $\ell_{\cB}(f)$. We should also mention that the same function was considered in the paper \cite{Salomaa:2003} under the name \emph{depth}.

For any subset $\mathcal{C}\subseteq \Tran_n$ we define the following function:
\begin{equation}\label{eq:set_compl}
\ell(\mathcal{C}) \bydef \max_{\cB\subseteq \mathcal{C},f\in\langle \cB \rangle}\ell_{\cB}(f).
\end{equation}

The function $\ell(\mathcal{C})$ can be interpreted as the worst-case complexity of the functions from $\mathcal{C}$. In the paper \cite{sokolovskii:1976:en} Sokolovskii shows that:
\begin{equation*}
    \begin{aligned}
        \binom{n-1}{\floor{\frac{n-1}{2}}} < \ell(\Tran_n) < n^{\frac{n}{2}(1+o(1)))},\\
        \mathrm{e}^{\sqrt{n\ln n}(1+o(1))} < \ell(\Sym_n) < n!^{\frac12(1+o(1))},
    \end{aligned}
\end{equation*}
as $n\to\infty$. It is worth mentioning that the lower bound for $\ell(\Sym_n)$ follows from the asymptotic estimate of the maximum order of the permutations from $\Sym_n$ called Landau's function \cite{landau}. The stronger result for $\ell(\Sym_n)$ follows from \cite{Babai:2006}.
The author considers only \emph{closed} sets $\mathcal{C}$ (i.e., $\langle\mathcal{C}\rangle=\mathcal{C}$), which  are subgroups of $\Sym_n$. For any subgroup $G$ of $\Sym_n$ the \emph{directed diameter} $\mathrm{diam}^{+}(G)$ of the group $G$ is defined as follows:
\begin{equation*}
    \mathrm{diam}^{+}(G) \bydef \max_{f\in G, \langle\cB\rangle = G} \ell_\cB(f).\\
\end{equation*}

It is easily shown that $\ell(\Sym_n)=\max_{G}\mathrm{diam}^{+}(G)$, where $G$ ranges over all subgroups of $\Sym_n$. From the results of \cite{Babai:2006} it follows that
\begin{equation}\label{eq:sym_compl}
    \ell(\Sym_n) = \mathrm{e}^{\sqrt{n\ln n}(1+o(1))}\text{ as } n\to\infty.\\
\end{equation}

We are now ready to state the first of the two main results of this paper.
\begin{theorem}\label{th:tran}
We have $\ell(\Tran_n) = 2^n\mathrm{e}^{\sqrt{\frac{n}{2}\ln n}(1+o(1))}$ as $n\to\infty$.
\end{theorem}

As we mentioned before, Sokolovskii discovered (see \cite{sokolovskii:1976:en}) the relationship between functions $\ell(\Tran_n)$ and $\ell(n,k)$. He proved in particular that
\begin{equation}\label{th:sok:up}
    \ell(n,k) \le (k-1)\ell(\Tran_n).\\
\end{equation}

The \emph{binary entropy} function denoted by $H_2(x)$ is defined as follows:
\begin{equation*}
    H_2(x) \bydef -x\log_2 x - (1-x)\log_2 (1-x),\text{ where }x\in (0,1).\\
\end{equation*}

Combining inequalities (\ref{th:sok:low1}), (\ref{th:sok:low2}), and (\ref{th:sok:up}) with theorem \ref{th:tran}  the second main result of the paper can be proved.
\begin{theorem}\label{th:diag}
We have $\log_2 \ell(n,k) \sim \varphi(a) n$ as $n\to\infty$ and $k/n\to a\in (0,1)$, where $\varphi(a)=H_2(a)$ if $a < 1/2$ and $\varphi(a)=1$ if $a \ge 1/2$ (see Fig. \ref{fg:phi}).
\end{theorem}

\begin{figure}\center
\begin{tikzpicture}[std-opt]
\tikzstyle{every text node part}=[font=\footnotesize]
\begin{scope}[xscale=3,yscale=2]
    \draw[domain=0.01:0.5,smooth]
        (0,0) coordinate (O) --
        plot (\x,{-\x*log2(\x)-(1-\x)*log2(1-\x)}) --
        (0.5,1) -- (1,1);
    \draw     (O)      node[point,label=below right:0] {} --
              +(0.5,0) node[point,label=below:$\frac{1}{2}$] (k) {} --
              (1,0) coordinate (X)     node[point,label=below:1] {};
    \draw     (O) --
              (0,1)        node[point,label=left:$1$] (Hk) {};
    \draw[dashed] (Hk) -- (Hk -| k) node[point] {} -- (k);
\end{scope}
    \draw (O) -- +(-0.7,0);
    \draw (O) -- +(0,-0.7);
    \draw[->] (X) -- +(1,0) node[right]  {$a$};
    \draw[->] (Hk) -- +(0,1) node[above] {$\varphi(a)$};
\end{tikzpicture}
\caption{Function $\varphi(a)$}\label{fg:phi}
\end{figure}
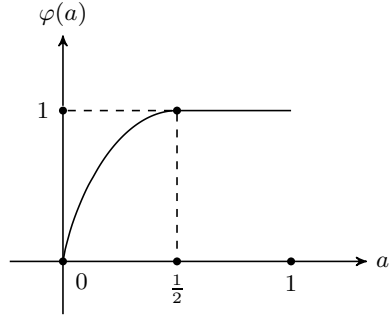

\section{Proofs of the Main Results}

Before we proceed to the formal proofs of the main results, let us give some definitions and state some useful lemmas first. Consider the set $\Tran_n^{(k)}$ of all bijections ${f\colon D \to D'}$ such that $D,D'\subseteq \Omega_n$ and $\abs{D}=\abs{D'}=k$. Suppose ${\cB\subseteq \Tran_n}$, ${f\in \Tran_n^{(k)}}$, ${f\colon D\to D'}$, and there is a map ${g\in \langle \cB \rangle}$ such that $f=g|_{D}$. Then we denote by $\ell_B(f)$ the minimum of  $\ell_B(g)$ over all such maps $g$, or $0$ if there are no such maps. The value $\ell_\cB(f)$ is also called the \emph{complexity} of $f$ over~$\cB$. Consider the following function:
\[\ell(\Tran_n^{(k)}) \bydef \max_{\cB\subseteq \Tran_n, f\in \Tran_n^{(k)}} \ell_B(f).\]

If $f(D)=D'$, then we say that \emph{$f$ transforms $D$ into $D'$} and write $D \xrightarrow{f} D'$.
For any set of maps $\cB\subseteq \Tran_n$ the \emph{$k$-graph over $\cB$} is the directed graph $G$ (loops and multiple edges are permitted\footnote{Sometimes such graphs are called \emph{pseudographs}.}) with the set of vertices
\[V(G)\bydef \{D \mid D \subseteq \Omega_n, \abs{D}=k \},\]
the set of arcs
\[E(G)\bydef \{f\in \Tran_n^{(k)} \mid f = g|_D\text{ for some } g\in \mathcal{B}\text{ and }D\in V(G) \},\]
and every arc $f$ goes from the vertex $D$ to the vertex $D'$ whenever $D \xrightarrow{f} D'$.

A \emph{walk} from the vertex $D$ to the vertex $D'$ in the $k$-graph $G$ is a sequence of vertices and arcs $\mathbf{w}=D_0,f_1,D_1, \ldots, f_\ell,D_\ell$
such that $D_0=D$, $D_\ell=D'$ and the arc $f_i$ goes from the vertex $D_{i-1}$ to the vertex $D_i$ for $i=1,\ldots,\ell$,
or, in terms of maps,\[D_0\xrightarrow{f_1} D_1 \xrightarrow{f_2}\cdots \xrightarrow{f_\ell} D_\ell.\] We often omit
vertices in walks and write simply $\mathbf{w}=f_1,\ldots, f_\ell$. The number $\ell$ is called the \emph{length} of the walk $\mathbf{w}$ and is denoted by $\len{\mathbf{w}}$. By a \emph{subwalk} of $\mathbf{w}$ we mean a subsequence $f_i,f_{i+1},\ldots,f_j$, $1 \le i < j \le \ell$.

For any walk $\mathbf{w}=f_1,\ldots,f_\ell$ from $D$ to $D'$ consider the map $[\mathbf{w}]\colon D\to D'$, where $[\mathbf{w}] = f_1 \ldots f_\ell$ (the composition of the maps $f_1,\ldots, f_\ell$). Two walks $\mathbf{w}$ and $\mathbf{w}'$ are called \emph{equivalent} if $[\mathbf{w}]=[\mathbf{w}']$. For a \emph{closed} walk $\mathbf{w}$, which starts and ends in the same vertex $D$, the map $[\mathbf{w}]$ is a permutation of $D$. For any closed walk~$\mathbf{w}$, by definition, put
\[\mathbf{w}^k \bydef \underbrace{\mathbf{w}, \ldots, \mathbf{w}}_k,\quad k\in \N.\]
It is readily seen that $\mathbf{w}^k$ is also a closed walk and $[\mathbf{w}^k]=[\mathbf{w}]^k$.

The next lemma is an immediate consequence of the previous definitions.
\begin{lemma}\label{lm:tran_compl}
Given a basis $\cB\subseteq \Tran_n$ and  a map $f\in\Tran_n^{(k)}$ such that $f\colon D\to D'$ is a restriction of some map from $\langle\cB\rangle$. Consider the $k$-graph $G$ over $\cB$. Then  $\ell_\cB(f)$ is the length of a shortest walk $\mathbf{w}$ in $G$ from $D$ to $D'$ such that $[\mathbf{w}]=f$.
\end{lemma}

We say that a vertex $D$ is \emph{reachable from} a vertex $D'$ in a $k$-graph $G$ if there is a walk in $G$ from $D$ to $D'$. Vertices $D$ and $D'$ are called \emph{mutually reachable} if $D$ is reachable from $D'$ and $D'$ is reachable from $D$. A $k$-graph is called \emph{strongly connected} if all its vertices are mutually reachable. Obviously, mutual reachability is an equivalence relation on vertices and  it partitions the set of vertices $V(G)$ into equivalence classes $V(G) = V_1 \cup \ldots \cup V_r$. Subgraphs $G_1,\ldots,G_r$ induced by $V_1,\ldots,V_r$ are called \emph{strongly connected components} of $G$. Evidently, every strongly connected component is strongly connected.

\begin{lemma}\label{lm:tran_main}
For any walk $\mathbf{w}$ in a strongly connected $k$-graph $G$ over $\mathcal{B}\subseteq \Tran_n$ there is an equivalent walk $\mathbf{w}'$ such that $\len{\mathbf{w}'} < 2\abs{V(G)}\cdot(\ell(\Sym_k)+1)-1$.
\end{lemma}
\begin{proof}
 Given a walk $\mathbf{w}= D_0, f_1, D_1,\ldots,f_\ell,D_\ell$ in the $k$-graph~$G$. For any vertex~$D$ in $G$ we define a walk~$\mathbf{w}_D$, equivalent to $\mathbf{w}$, called a \emph{$D$-saturation} of $\mathbf{w}$,  as follows. For every vertex $D_i$, $0\le i \le \ell$, we consider two paths\footnote{A \emph{path} is a walk in which all vertices and edges are distinct.}:  $\mathbf{p}_{D_i\to D}$ from $D_i$ to $D$ and $\mathbf{p}_{D\to D_i}$ from $D$ to $D_i$. Connecting them, we obtain the closed walk $\mathbf{c}_i=\mathbf{p}_{D_i\to D},\mathbf{p}_{D\to D_i}$. Since $G$ is strongly connected, it follows that these two paths exist and $\len{\mathbf{p}_{D\to D_i}}$,  $\len{\mathbf{p}_{D\to D_i}}$ are bounded by $\abs{V(G)}-1$.
For every closed walk $\mathbf{c}_i$ we consider the permutation $\pi_i=[\mathbf{c}_i]$ of the set~$D_i$. Let $m_i$ be the \emph{order} of $\mathbf{c}_i$, i.e., the smallest positive integer $m$ such that $\pi_i^m  = \id_{D_i}$ (where $\id_M$ denotes the identity map on $M$). Finally, by definition, put
\begin{equation}\label{eq:satur}
\mathbf{w}_D \bydef \mathbf{c}_0^{m_0}, f_1, \mathbf{c}_1^{m_1},\ldots, f_\ell, \mathbf{c}_\ell^{m_\ell}.
\end{equation}
It now follows that
\[[\mathbf{w}_D] = \pi_0^{m_0} f_1 \pi_1^{m_1} \ldots f_\ell \pi_\ell^{m_\ell} =
    \id_{D_0} f_1 \id_{D_1} f_2 \ldots f_\ell \id_{D_\ell} = f_1 f_2 \ldots f_\ell =
    [\mathbf{w}],\]
and we see that the $D$-saturation $\mathbf{w}_{D}$ is equivalent to the walk $\mathbf{w}$.

Consider all the occurrences of the vertex $D$ in the walk $\mathbf{w}_D$. These occurrences partition  $\mathbf{w}_D$ into subwalks, i.e., $\mathbf{w}_D = \mathbf{w}_0,\mathbf{w}_1,\ldots,\mathbf{w}_s,\mathbf{w}_{s+1}$, where $\mathbf{w}_0$ is the subwalk from the begin to the first occurrence of $D$, $\mathbf{w}_{s+1}$ is the subwalk from the last occurrence of $D$ to the end, and the closed subwalks $\mathbf{w}_1,\ldots,\mathbf{w}_s$ connect successive occurrences of $D$. Using (\ref{eq:satur}) and recalling that $\mathbf{c}_i=\mathbf{p}_{D_i\to D},\mathbf{p}_{D\to D_i}$, where  $\len{\mathbf{p}_{D\to D_i}}$ and $\len{\mathbf{p}_{D\to D_i}}$ are bounded by $\abs{V(G)}-1$, we have $\len{\mathbf{w}_0}\le \abs{V(G)}-1$, $\len{\mathbf{w}_{s+1}}\le \abs{V(G)}-1$, and $\len{\mathbf{w}_i}\le 2\abs{V(G)}-1$ for $i=1,\ldots,s$.
Let $\pi_1 = [\mathbf{w}_1],\ldots,\pi_s = [\mathbf{w}_s]$. Consider the set $\cB=\{\pi_1,\ldots,\pi_s\}$ and the permutation $\pi = \pi_1\ldots \pi_s\in \langle \cB \rangle$. Now note that $\pi,\pi_1,\ldots,\pi_s$ are permutations of the same $k$-element set $D$. Thus, taking into account (\ref{eq:set_compl}), we obtain $\pi = \pi_{i_1}\ldots\pi_{i_r}$, where $r\le \ell(\Sym_k)$.

Finally, let $\mathbf{w}'\bydef \mathbf{w}_0, \mathbf{w}_{i_1}, \ldots, \mathbf{w}_{i_r}, \mathbf{w}_{s+1}$. Then we get
\[[\mathbf{w}'] = [\mathbf{w}_0]\pi_{i_1}\ldots\pi_{i_r}[\mathbf{w}_{s+1}] = [\mathbf{w}_0]\pi_{1}\ldots\pi_{r}[\mathbf{w}_{s+1}] = [\mathbf{w}_D]=[\mathbf{w}],\]
i.e., $\mathbf{w}'$ is equivalent to $\mathbf{w}$. Moreover, we have
\[\len{\mathbf{w}'} \le 2(\abs{V(G)} - 1) + (2\abs{V(G)}-1)\cdot \ell(\Sym_k) < 2\abs{V(G)}\cdot (\ell(\Sym_k)+1)-1.\]
The lemma is proved.\qed
\end{proof}

In the previous lemma we deal with strongly connected $k$-graphs only. More general case is considered in the next lemma.
\begin{lemma}\label{lm:tran_main2}
For any walk $\mathbf{w}$ in a $k$-graph $G$ over $\mathcal{B}\subseteq \Tran_n$ there is an equivalent walk $\mathbf{w}'$ such that $\len{\mathbf{w}'} < 2\abs{V(G)}\cdot(\ell(\Sym_k)+1)$.
\end{lemma}
\begin{proof}
Consider an arbitrary walk $\mathbf{w}$ in $G$. It is readily seen that this walk can be represented as $\mathbf{w} = \mathbf{w}_1, f_1, \mathbf{w}_2,\ldots,f_{s-1}\mathbf{w}_s$, where every subwalk $\mathbf{w}_i$ belongs completely to one strong component $G_i$ of $G$ and all the components $G_1,\ldots,G_s$ are different. On the other hand, from Lemma \ref{lm:tran_main} it follows that for any walk $\mathbf{w}_i$, $1\le i \le s$, there exists an equivalent walk $\mathbf{w}_i'$ such that ${\len{\mathbf{w}_i'} < 2\abs{V(G_i)}\cdot (\ell(\Sym_k)+1)-1}$. Then we let $\mathbf{w}'=\mathbf{w}_1',f_1,\mathbf{w}_2',\ldots,f_{s-1},\mathbf{w}_{s}'$ and obtain that $[\mathbf{w}']=[\mathbf{w}]$. Moreover, we have
\[\len{\mathbf{w}'} = \sum_{i=1}^{s}\len{\mathbf{w}_i'}+s-1 < \left(\sum_{i=1}^s 2\abs{V(G_i)}\right)\cdot\left(\ell(\Sym_k)+1\right) \le 2\abs{V(G)}\cdot (\ell(\Sym_k) + 1).\]
This proves the lemma.\qed
\end{proof}

Consider an arbitrary basis $\cB\subseteq\Tran_n$. It is clear that for the $k$-graph $G$ over $\cB$ we have $\abs{V(G)}=\binom{n}{k}$. Therefore from Lemmas  \ref{lm:tran_compl} and \ref{lm:tran_main2} it follows that
\begin{equation}\label{eq:tran_up2}
\ell(\Tran_n^{(k)}) < 2\binom{n}{k}(\ell(\Sym_k)+1).
\end{equation}

Combining this fact with equality (\ref{eq:sym_compl}), we obtain the following
\begin{lemma}\label{lm:tran_main3}
We have $\ell(\Tran_n^{(k)}) < \binom{n}{k}\mathrm{e}^{\sqrt{k\ln k}(1+o(1))}$ as $n,k\to\infty$.
\end{lemma}

Consider a basis $\cB\subseteq \Tran_n$ and a map $f\in\langle\cB\rangle$. Let $f=f_1\ldots f_\ell$ be a shortest representation of $f$ over $\cB$, i.e., $\ell=\ell_\cB(f)$. Thus we have
\[D_0\xrightarrow{f_1} D_1 \xrightarrow{f_2} \cdots \xrightarrow{f_\ell} D_\ell,\]
where $D_i\bydef \Omega_n f_1\ldots f_i$. Suppose $k_i\bydef\abs{D_i}$ for $i=0,\ldots,\ell$; then we obtain
\[k_0=\ldots=k_{i_1}=r_0>k_{i_1+1}=\ldots=k_{i_2}=r_1>\ldots > k_{i_s + 1}=\ldots=k_\ell=r_s.\]
Therefore we get $f=g_0 f_{i_1} g_1 \ldots g_{s-1} f_{i_s} g_s$, where $g_i\in \Tran_n^{(r_i)}$, $0\le i \le s$.
Thus it is easily shown that $\ell_\cB(f) = s + \ell_\cB(g_0)+\cdots + \ell_\cB(g_s)$. It is clear that $s\le n-1$. Therefore, we have
\begin{equation*}
\ell_\cB(f) \le n-1 + \ell(\Tran_n^{(r_0)}) +\cdots + \ell(\Tran_n^{(r_s)})  < n \max_{0\le i \le s} \left\{\ell(\Tran_n^{(r_i)}) + 1\right\}.
\end{equation*}
Finally, we obtain
\begin{equation}\label{eq:tran_up}
\ell(\Tran_n) <  n \max_{1\le k \le n} \left\{\ell(\Tran_n^{(k)}) + 1\right\}.
\end{equation}

\begin{lemma}\label{lm:tran_up}
We have $\ell(\Tran_n) \le 2^n\mathrm{e}^{\sqrt{\frac{n}{2}\ln n}(1+o(1))}$ as $n\to\infty$.
\end{lemma}
\begin{proof}
Using  (\ref{eq:sym_compl}), (\ref{eq:tran_up2}), and (\ref{eq:tran_up}), we obtain $\ln \ell(\Tran_n) \le \ln n + \ell(n)$, where
\[\ell(n) \bydef \max_{1\le k\le n}\left\{\ln \binom{n}{k}+\varphi(k)\right\},\]
and $\varphi(k)$ is a function such that $\varphi(k)\sim \sqrt{k\ln k}$ as $k\to\infty$.

Recall that the function $\ln \binom{n}{k}$ achieves its maximum value at the point ${k_n = \ceil{n/2}}$ when $n$ is fixed. Suppose $\ln\binom{n}{k}+\varphi(k)$ achieves its maximum value at $k_n'=k_n+h_n$, i.e., $\ell(n)=\ln\binom{n}{k_n'}+\varphi(k_n')$. We claim that $h_n/n\to 0$ as $n\to\infty$. Indeed, in the converse case, we can take $\varepsilon\in (0,1/2)$  such that $\abs{h_n/n}\ge \varepsilon$ holds for an infinite sequence of indexes $n$.
Further, since we have $\abs{h_n/n}\in [\varepsilon,1/2]$ for this sequence; then it has an infinite subsequence $n_1,n_2,\ldots, n_i,\ldots$ such that ${h_{n_i}/n_i\to a\in [-1/2,-\varepsilon]\cup[\varepsilon,1/2]}$ as $i\to\infty$.   On the other hand, it is well known that
\begin{equation}\label{eq:entropy}
\frac{\ln\binom{n}{m}}{n}\to H(p)
\end{equation}
as $n,m\to\infty$ and $m/n\to p\in [0,1]$, where\footnote{Here we assume that $0\cdot\ln 0 = 0$.}
\[H(p) \bydef -p\ln p - (1-p)\ln(1-p)\]
is the \emph{entropy function}. Since $\varphi(k_{n_i})=o(n_i),\varphi(k_{n_i}')=o(n_i)$, ${k'_{n_i} / n_i \to 1/2+a}$, and ${k_{n_i} / n_i \to 1/2}$ as $i\to\infty$, we obtain:
\[
    \lim_{i\to\infty} \frac{\ln \binom{n_i}{k'_{n_i}}+\varphi(k'_{n_i})}{n_i}=H\left(\frac12 + a\right) \ge
    \lim_{i\to\infty} \frac{\ln \binom{n_i}{k_{n_i}}+\varphi(k_{n_i})}{n_i}=H\left(\frac12\right).
\]
The latter contradicts the fact that the function $H(p)$ achieves its maximum value at the point $1/2$ only (see Fig. \ref{fg:entropy}). This contradiction proves that $h_n/n\to 0$ as $n\to\infty$.
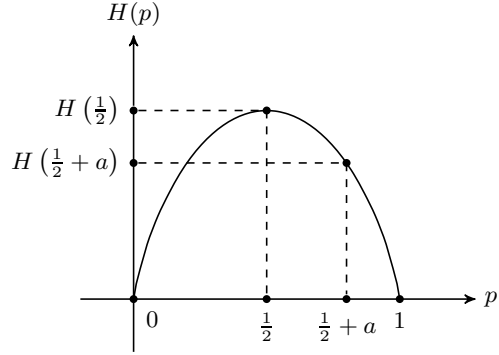
\begin{figure}\center
\begin{tikzpicture}[std-opt]
\tikzstyle{every text node part}=[font=\footnotesize]
\begin{scope}[xscale=3.5,yscale=2.5]
    \draw[domain=0.01:0.99,smooth]
        (0,0) coordinate (O) --
        plot (\x,{-\x*log2(\x)-(1-\x)*log2(1-\x)}) --
        (1,0) coordinate (X);
    \draw     (O)      node[point,label=below right:0] {} --
              +(0.5,0) node[point,label=below:$\frac{1}{2}$] (k) {} --
              +(0.8,0) node[point,label=below:$\frac{1}{2}+a$] (k') {} --
              (X)      node[point,label=below:1] {};
    \draw     (O) --
              (0,.721928)  node[point,label=left:$H\left(\frac12 + a\right)$] (Hk') {} --
              (0,1)        node[point,label=left:$H\left(\frac12\right)$] (Hk) {};
    \draw[dashed] (Hk) -- (Hk -| k) node[point] {} -- (k);
    \draw[dashed] (Hk') -- (Hk' -| k') node[point] {} -- (k');
\end{scope}
    \draw (O) -- +(-0.7,0);
    \draw (O) -- +(0,-0.7);
    \draw[->] (X) -- +(1,0) node[right]  {$p$};
    \draw[->] (Hk) -- +(0,1) node[above] {$H(p)$};
\end{tikzpicture}
\caption{Entropy function $H(p)$}\label{fg:entropy}
\end{figure}

Further, since the function $\ln \binom{n}{k} + \varphi(k_n)$ achieves its maximum value at $k=k'_n$ and the function $\ln \binom{n}{k}$ at $k=k_n$ when $n$ is fixed; then we obtain
\begin{equation}\label{eq:1}
    0 \le \ln \binom{n}{k'_n} + \varphi(k'_n) - \left(\ln \binom{n}{k_n} + \varphi(k_n)\right) \le \varphi(k'_n)-\varphi(k_n).
\end{equation}

Since $h_n / n \to 0$, we see that $k_n\sim k_n'$ and $\sqrt{k_n\ln k_n}\sim\sqrt{k'_n\ln k'_n}$ as $n\to\infty$. From $\varphi(k_n)\sim\sqrt{k_n\ln k_n}$ and $\varphi(k'_n)\sim\sqrt{k'_n\ln k'_n}$ it follows that $\varphi(k_n)\sim\varphi(k'_n)$. Hence $\varphi(k'_n) - \varphi(k_n)=o(\sqrt{k_n\ln k_n})=o(\sqrt{n\ln n})$ as $n\to\infty$.

Thus, recalling that $\binom{n}{\ceil{n / 2}}\sim\sqrt{\frac{2}{\pi n}}\cdot 2^n$ as  $n\to\infty$, from~(\ref{eq:1}) it follows that
\[
    \ln \binom{n}{k'_n} + \varphi(k'_n) = \ln \binom{n}{k_n} + \varphi(k_n) + o(\sqrt{n \ln n})=n\ln 2+\sqrt{\frac{n}2 \ln n} + o(\sqrt{n \ln n}).
\]
Therefore
\[\ell(n)=n\ln 2+\sqrt{\frac{n}2 \ln n}\cdot(1+o(1))\]
and recalling that ${\ln \ell(\Tran_n) \le \ln n + \ell(n)}$,  we obtain
\[\ell(\Tran_n) \le n e^{\ell(n)}= 2^n e^{\sqrt{\frac{n}{2} \ln n}(1+o(1))}\text{ as }n\to\infty.\]
This completes the proof.
\qed
\end{proof}

Each semiautomaton ${\gA=(A,Q,\tranf)}$ induces the transformation semigroup $\Tran(\gA)$ acting on the set of states $Q$ in the following way. For every word $\alpha\in A^*$, let $T_\alpha\colon Q\to Q$ be the map $q\mapsto \tranf(q,\alpha)$.
Then by definition, put
\[\Tran(\gA)=\{T_\alpha\mid \alpha\in A^*\}.\]

It is obvious that $\Tran(\gA)=\langle \cB \rangle$, where $\cB=\{T_a\mid a\in A\}$. Moreover, if $f\in\Tran(\gA)$, then $\ell_\cB(f)$ is equal to the length of a shortest word $\alpha\in A^*$ such that $f=T_\alpha$.

In Lemma \ref{lm:tran_up} we obtain an upper bound on the function $\ell(\Tran_n)$. The next lemma shows that this bound is in some sense exact.

\begin{lemma}\label{lm:tran_low}
We have $\ell(\Tran_n) \ge 2^n\mathrm{e}^{\sqrt{\frac{n}{2}\ln n}(1+o(1))}$ as $n\to\infty$.
\end{lemma}
\begin{proof}
For each $n$ and $k < n$ consider a semiautomaton $\gA=(A,Q,\tranf)$ such that ${Q=\{1,\ldots,n\}}$, ${A=\{1,\ldots,m\}}$, where $m=\binom{n-1}{k}$; and the transition function $\delta$ is defined as follows. First we take in some order all  $k$-element subsets of the set~$\{1,\ldots,n-1\}\subseteq Q$:
\[D_1=\{q^{(1)}_1,\ldots,q^{(1)}_k\}, \dots, D_m=\{q^{(m)}_1,\ldots,q^{(m)}_k\}.\]
Further, we choose a permutation $\pi\in\Sym_k$ of the maximum order,
and define the transition function such that $\tranf(q^{(i)}_j,i) = q^{(i+1)}_j$ and $\tranf(q^{(m)}_j,m)=q^{(1)}_{\pi(j)}$ for $i=1,\ldots, m-1$; ${j=1,\ldots,k}$. Moreover, we let $\tranf(q,i)=n$ whenever $q\notin D_i$ for $i=1,\ldots, m$.

It is not hard to see that we have
\begin{equation}\label{eq:tran_low}
D_1\xrightarrow{\ 1\ } D_2 \xrightarrow{\ 2\ } \cdots \xrightarrow{m-1} D_m \xrightarrow{\ m\ } D_1.
\end{equation}
Furthermore, we claim that we have $\tranf(D_1,\alpha)=D_1$ iff $\alpha=(12\ldots m)^s$, $s\ge 0$. Indeed, if $\alpha=(12\ldots m)^s$; then from (\ref{eq:tran_low}) we obtain $\tranf(D_1,\alpha)=D_1$. Suppose we have $\tranf(D_1,\alpha)=D_1$ for some word $\alpha=a_1\ldots a_\ell\in A^{*}$. Consider the sequence $D_1',\ldots,D_\ell'$, where $D_1'=D_\ell'=D_1$ and $D_{i+1}'=\tranf(D_i',a_i)$ for $i=1,\ldots,\ell-1$. Let us show that
\begin{equation}\label{eq:tran_low2}
a_1=1,a_2=2,\ldots,a_m=m, a_{m+1}=1, a_{m+2}=2,\ldots,a_\ell=m.
\end{equation}
Assume the converse, and let $i$ be the smallest index such that condition (\ref{eq:tran_low2}) does not hold for $a_i$. Then it is readily seen that
\[D_1'=D_1,D_2'=D_2,\ldots,D_i'=D_i\text{ and }n\in D_{i+1}'\ne D_{i+1}.\]
Further, since $\tranf(n,a)=n$ for all $a\in A$; then we get $n\in\tranf(D_i',a_{i+1}\ldots a_\ell)=D_\ell'$, and hence $n\in D_1\subseteq \{1,\ldots,n-1\}$. This contradiction proves condition (\ref{eq:tran_low2}) and we obtain $\alpha=(12\ldots m)^s$ for some $s\ge 0$.

Let $r_k$ be the order of the previously defined permutation $\pi\in \Sym_k$. Hence $r_k$ is Landau's function \cite{landau}, i.e., the maximum order of an element of $\Sym_k$, and we get $r_k=\mathrm{e}^{\sqrt{k\ln k}(1+o(1))}$ as $k\to\infty$. Consider the map $f\colon q\mapsto \tranf(q,(12\ldots m)^{r_k-1})$. Since $f(D_1)=D_1$; then  for each word $\alpha\in A^*$ such that $f=T_{\alpha}$ we have $\alpha=(12\ldots m)^s$, and it is not hard to see that $f|_{D_1}=\pi^s$. Therefore we have $s\ge r_k-1$ and
$\ell_\cB(f)\ge \abs{\alpha}=ms \ge  \binom{n-1}{k}(r_k-1)$, where $\cB=\{T_a\mid a\in A\}$. Finally, if we let $k=\floor{n/2}$, then we obtain the inequality
\[\ell_\cB(f) \ge \binom{n-1}{k}(r_k-1)=2^n\mathrm{e}^{\sqrt{\frac{n}{2}\ln n}(1+o(1))}\text{ as }n\to\infty.\]
This completes the proof of the lemma.\qed
\end{proof}

Now we can prove the first of the two main results of this paper.
\begin{proof}[of Theorem \ref{th:tran}]
The result follows from Lemmas \ref{lm:tran_up} and \ref{lm:tran_low}.\qed
\end{proof}

Before we give the proof of Theorem \ref{th:diag}, we introduce the following definitions and notions.

A \emph{partition} of a set $S$ is a set $\pi=\{B_1,\ldots,B_m\}$ of pairwise disjoint non-empty subsets $B_i\subseteq S$ (called \emph{blocks}) such that $\cup_i B_i = S$. We say that a partition $\pi'$ is a \emph{refinement} of a partition $\pi$ and write $\pi' \le \pi$ if every element of $\pi'$ is a subset of some element of $\pi$. It is easily shown that the set of all partitions of $S$ is a partially ordered set with respect to the relation ``$\le$''. It has the least element (called \emph{discrete} partition), which contains $\abs{S}$ singleton  blocks, and the greatest element (called \emph{trivial} partition), which contains one $\abs{S}$-element block.

Given a finite automaton $\gA=(A,Q,B,\tranf,\outf)$ and a subset of states $S\subseteq Q$, the \emph{initial state uncertainty} (with respect to $\gA$ and $S$) after applying input word $\alpha$ is a partition $\pi_\alpha$ of $S$ such that two states $q, q'\in S$ are in the same block iff $\outf(q,\alpha)=\outf(q',\alpha)$. Informally speaking, the initial state uncertainty describes what we know about the initial state $q_0\in S$ of the automaton $\gA$ after applying the input word $\alpha$. From the definition of a PDS for $S$ in $\gA$ it follows that an~input word $\alpha$ is a PDS iff the partition $\pi_\alpha$ is discrete. Moreover, it is easy to prove that for every $\alpha,\beta\in A^*$ the partition $\pi_{\alpha\beta}$ is a refinement of $\pi_{\alpha}$.

\begin{proof}[of Theorem \ref{th:diag}]
Given an $n$-state automaton $\gA$ and a $k$-element subset $S$ of its states. Let $\alpha=a_1\dots a_\ell$ be a minimum length PDS for the subset~$S$. For each $i\in \{0,1,\dots,\ell\}$ we consider the two values ${k_i=\abs{\tranf(S,a_1\ldots a_i)}}$ and $r_i=\abs{\pi_{a_1\ldots a_i}}$. It is clear that $k_0 \ge k_1 \ge \dots \ge k_l$ and $r_0 \le r_1 \le \dots \le r_l$. Let $i_1,  \dots,  i_m$ be the increasing sequence of all indexes $i\in\{1,\dots,\ell\}$ such that $k_{i-1} > k_i$ or $r_{i-1} < r_i$. Since $\alpha$ is a minimum length PDS, then $r_{\ell-1} < r_\ell$ and $i_m=\ell$. Let also $i_0=0$. Hence the word $\alpha$ can be represented as $\alpha=\alpha_1 a_{i_1} \dots \alpha_m a_{i_m}$. Moreover, it is readily seen that $m\le 2(n-1)$.

Further, for each $j \in\{1,\dots,m\}$ there exists an input word $\alpha_j'$ such that ${T_{\alpha_j}|_{S_{j-1}} = T_{\alpha_j'}}|_{S_{j-1}}$ and $\abs{\alpha_j'} \le \ell(\Tran_n^{(p_j)})$, where $S_j = \tranf(S,a_1\ldots a_{i_j})$,  $p_j=k_{i_{j-1}}=k_{i_{j-1}+1}=\dots=k_{i_j-1}$. We claim that the word $\alpha'=\alpha_1' a_{i_1}\ldots \alpha_m' a_{i_m}$ is also a PDS for $S$. Indeed, in the converse case, there exist two states $q_1,q_2\in S$ such that $\outf(q_1,\alpha') = \outf(q_2,\alpha')$. Since $\alpha$ is a PDS, we obtain ${\outf(q_1,\alpha) \ne \outf(q_2,\alpha)}$. Let $j$ be the minimum index such that
\[\outf(q_1,\alpha_1 a_{i_1} \ldots \alpha_j a_{i_j})\ne \outf(q_2,\alpha_1 a_{i_1} \ldots \alpha_j a_{i_j}).\]
Then from $r_{i_j} > r_{i_j-1}=\dots=r_{i_{j-1}}$ and the minimality of the index $j$ it follows that
\[\outf(q_1,\alpha_1 a_{i_1} \ldots \alpha_j) = \outf(q_2,\alpha_1 a_{i_1} \ldots \alpha_j).\]
Therefore for the states $q_1'=\tranf(q_1,\alpha_1 a_{i_1} \ldots \alpha_j)$, $q_2'=\tranf(q_2,\alpha_1 a_{i_1} \ldots \alpha_j)$ we get $\outf(q_1',a_{i_j})\ne \outf(q_2',a_{i_j})$. At the same time since $T_{\alpha_1 a_{i_1} \ldots \alpha_j}|_S=T_{\alpha_1' a_{i_1} \ldots \alpha_j'}|_S$, we have $q_1'=\tranf(q_1,\alpha_1' a_1 \ldots \alpha_j')$, $q_2'=\tranf(q_2,\alpha_1' a_1 \ldots \alpha_j')$ and we finally obtain
\[\outf(q_1,\alpha_1' a_1 \ldots \alpha_j' a_{i_j}) \ne \outf(q_2,\alpha_1' a_1 \ldots \alpha_j' a_{i_j}).\]
Therefore the word $\alpha'$ distinguishes the states $q_1$, $q_2$ and hence is a PDS for $S$. Moreover, we have
$\abs{\alpha'}\le m + \abs{\alpha_1'} + \cdots + \abs{\alpha_m'}\le m + \ell(\Tran^{(p_1)}_n)+\dots+\ell(\Tran^{(p_m)}_n)$ and therefore
\[\ell(n,k) < m\max_{1\le p\le k}\left\{\ell(\Tran^{(p)}_n) + 1\right\}.\]

Since the function $\ell(\Sym_k)$ is increasing; then from $k\le n$, $m\le 2(n-1)$, asymptotic equality (\ref{eq:sym_compl}),  and inequality~(\ref{eq:tran_up2}) it follows that
\begin{align}
\ell(n,k) & < \binom{n}{k}\mathrm{e}^{\sqrt{n\ln n}(1+o(1))}\quad\text{if }k \le \frac{n}{2};\\
\ell(n,k) & < 2^n\mathrm{e}^{\sqrt{n\ln n}(1+o(1))}\quad\text{if }k > \frac{n}{2}.
\end{align}

To conclude the proof, it remains to use inequalities (\ref{th:sok:low1}) and (\ref{th:sok:low2}) with asymptotic equality (\ref{eq:entropy}).\qed
\end{proof}

\section{Remarks and Related Work}
Despite the fact that the length of a shortest PDS is exponential in the worst case in the class of all Mealy automata there are a number of natural automata classes where it is much smaller. For example, for the class of linear automata it is only logarithmic~\cite{Cohn:1964} and for the class of automata with finite memory it is linear~\cite{Rystsov:1973:en} in the number of states. Moreover, if in a reduced automaton $\gA$ for each input symbol $a$ and for each pair of different states $q,q'$ such that $\tranf(q, a) = \tranf(q', a)$ we always have $\outf(q,a)\ne \outf(q',a)$ then every preset homing sequence (PHS) for $\gA$ is also a PDS for $\gA$~\cite{Rystsov:1973:en}. Hence using the classical result of Hibbard~\cite{Hibbard:1961} for PHSs it immediately follows that for any such $n$-state automaton a PDS always exists, can be efficiently computed, and the length of a shortest PDS is upper bounded by $\frac{n(n-1)}{2}$. Moreover, this upper bound is tight~\cite{Karacuba:1960:en,Hibbard:1961}. The class of such automata was investigated by the author in~\cite{Panteleev:2007:en} under the name multiply reduced automata. It is interesting to note that exactly the same class was considered in a recent paper~\cite{Gunicen:2014} under the name DMFSM where an $O(n^3)$ upper bound on the PDS length was obtained and an $O(n^2)$ upper bound was only conjectured.

\bibliographystyle{splncs03} 
\bibliography{abbrevs,automata}
 
\end{document}